\documentclass{article}
\UseRawInputEncoding 
 \usepackage{arxiv}
\usepackage{times}

\usepackage{color}
\usepackage{yfonts}
\usepackage{upgreek}
\usepackage{latexsym}
\usepackage{wrapfig}
\usepackage{supertabular}
\usepackage{subfigure}
\usepackage{graphicx}
\usepackage{amsmath} 
\usepackage{amsthm}
\usepackage{amssymb}
\usepackage{algorithm}
\usepackage{algorithmicx}
\usepackage{algpseudocode}

\usepackage{times}
\usepackage{multirow}
\usepackage{subfig}

\newtheorem{Corollary}{Corollary}

\begin{document}

\title{\LARGE \bf A Convex Optimization Approach to Learning Koopman Operators}

\author{M. Sznaier\thanks{The author is with the Department of 
Electrical and Computer Engineering, Northeastern University, Boston, MA 02115, email {\tt msznaier@coe.neu.edu}. This work was partially supported by NSF grants  CNS--1646121, CMMI--1638234, IIS--1814631 and ECCSÑ1808381  and AFOSR grant FA9550-19-1-0005.}
}



\newcommand{\argmin}[1]{\underset{#1}{\operatorname{argmin}}\;}
\newcommand{\sN}{\mathbf{N}}
\newcommand{\sZ}{\mathbf{Z}}
\newcommand{\sR}{\mathbf{R}}
\newcommand{\sC}{\mathbf{C}}
\newcommand{\deq}{\doteq}
\newcommand{\st}{\ensuremath{\colon}}
\newcommand{\set}[1]{\mathcal{#1}}
\newcommand{\greekset}[1]{\boldsymbol{#1}}
\newcommand{\vect}[1]{\mathbf{#1}}
\newcommand{\mat}[1]{\mathbf{#1}}
\newcommand{\greekvect}[1]{\boldsymbol{#1}}
\newcommand{\smax}[1]{\overline{\sigma}\left( #1 \right)}
\newcommand{\eigmax}[1]{\overline{\lambda}\left( #1 \right)}
\newcommand{\lp}[1]{\ell_{#1}}
\newcommand{\Lp}[1]{\mathcal{L}_{#1}}
\newcommand{\wk}[1]{e^{j\omega_{#1}}}

\newcommand{\diag}{\operatorname{diag}}
\newcommand{\trace}{\operatorname{Tr}}
\newcommand{\rank}{\operatorname{rank}}

\newtheorem{counterexample}{Counter--Example}
\newtheorem{coro}{Corollary}
\newtheorem{condition}{Condition}

\newtheorem{Algorithm}{Algorithm}
\newtheorem{Problem}{Problem}
\newtheorem{Definition}{Definition}
\newtheorem{Theorem}{Theorem}
\newtheorem{Lemma}{Lemma}
\newtheorem{Remark}{Remark}
\newtheorem{Example}{Example}
\newtheorem{Proposition}{Proposition}

\def\beq{\begin{equation}}
\def\eeq{\end{equation}}

\newcommand{\cb}[1]{{\color{blue} #1}}

\maketitle
\begin{abstract} 
Koopman operators provide tractable means of learning linear approximations of  non-linear dynamics. Many approaches have been proposed to find these operators,  typically  based upon approximations using an a-priori fixed 
 class of models. However, choosing appropriate models and bounding the approximation error is far from trivial. 
 Motivated by these difficulties, in this paper we propose an optimization based approach to  learning  Koopman operators from data. Our results show that the Koopman operator, the associated Hilbert space of observables and a suitable dictionary can be obtained by solving two rank-constrained semi-definite programs (SDP). While in principle these problems are NP-hard, the use of standard relaxations of rank leads to  convex SDPs.
\end{abstract}

\section{Introduction and motivation}\label{sec:intro}

Many scenarios  involve predicting the output of an unknown non-linear system based on past measurements and some a-priori information. 
 Recently, substantial interest has been devoted to the use of Koopman operator based methods  to solve this problem, as a tractable alternative to nonlinear identification. An excellent introduction to the topic is given in \cite{Mezic2013}, and more recent references can be found in 
 \cite{Lusch2018,Otto2019}.
  Given a non-linear discrete time system of the form:
\beq \label{eq:nonlinear} \begin{split}
\greekvect{\xi}_{k+1} & = f(\greekvect{\xi}_k) \; \text{where} \;
\greekvect{\xi}_k  = \begin{bmatrix} \mat{x}_{k-r+1}^T &\ldots & \mat{x}_{k}^T\end{bmatrix}^T, \; \mat{x}_j \in \mathbb{R}^n 
\end{split}
\eeq
 {let $\mathbb{H}$ denote a Hilbert space of functions  $\greekvect{\psi}(\greekvect{\xi})\colon \mathbb{R}^{nr} \to \mathbb{R}^{mr}$(the so called observables)}. The Koopman $\mathcal{K}$ operator acts on the elements of $\mathbb{H}$, by propagating their values one step into the future:
\beq \label{eq:nonlineark}
(\mathcal{K} \circ \greekvect{\psi})(\greekvect{\xi}_k) = (\greekvect{\psi}\circ f)(\greekvect{\xi}_{k}) =\greekvect{\psi}(\greekvect{\xi}_{k+1}) \eeq
$\mathcal{K}$ is a linear operator, albeit typically infinite dimensional. { When it has a countable set of eigenfunctions $\greekvect{\phi}_i(.)$ with eigenvalues $\mu_i$, the observables {$\greekvect{\psi}(.)$} can be propagated as follows.  Let $\mat{a}=\begin{bmatrix}a_1 \ldots \end{bmatrix}^T$ denote the coordinates of  $\greekvect{\psi}(.)$ in the basis spanned by $\greekvect{\phi}(.)$, that is
 \[  \greekvect{\psi}(.) = \sum a_i \greekvect{\phi}_i (.) \doteq \greekvect{\Phi}(.)\mat{a}, \; \text{where: }  \greekvect{\Phi}(.)=\begin{bmatrix} \greekvect{\phi}_1(.)\ldots \end{bmatrix} \]
 Then
\[ (\mathcal{K} \circ \greekvect{\psi})(.) = \sum a_i \mu_i \greekvect{\phi}_i (.) = \greekvect{\Phi}(.)\mat{M}\mat{a}, \text{where $\mat{M}$ = diag($\mu_i$)} \] }

In particular, if the state $ \greekvect{\xi} \in \text{span}\{ \greekvect{\phi}_i \}$, then $ \greekvect{\xi}_{k+1}= \greekvect{\Phi}(\greekvect{\xi}_k)\mat{M}\mat{a}$.
While this approach leads for to linear representations of \eqref{eq:nonlinear}, identifying the Koopman eigenfunctions from data is not trivial.

Extended Dynamical Mode Decomposition (EDMD) type approaches seek to identify approximations to Koopman operators over a restricted subspace, defined by the span of a given dictionary $\mathcal{D}(.) \doteq \begin{bmatrix}  \greekvect{\psi}_1(.) \ldots  \greekvect{\psi}_N(.) \end{bmatrix}$. 
In this subspace, the Koopman operator can then be approximated by a matrix $\mat{K}\in \mathbb{R}^{N \times N}$  that
propagates the coefficients of the expansion, that is, for $\greekvect{\psi}(.) = \mathcal{D}(.)\mat{a}$, then $(\mathcal{K} \circ \greekvect{\psi})(.) = \mathcal{D}(.)\mat{Ka}$. 
{
Typically, given experimental data $\mat{X} \doteq \begin{bmatrix} \greekvect{\xi}_1 & \greekvect{\xi}_2 & \ldots & \greekvect{\xi}_T \end{bmatrix}$, $\mat{K}$
  is found by minimizing the one-step prediction error over a set of  observables.  Specifically, this approach considers  $m$ observables
   $\greekvect{\psi}^{(j)}(.) \doteq  \mathcal{D}(.)\mat{a}_j$, each defined by a coordinate vector $\mat{a}_j$, and solves:
\beq \label{eq:EDMD} \mat{K} = \mathop{argmin} \limits_{K} \sum_{j=1}^m \sum_{k=1}^{T-1} \| [\mat{D}( \greekvect{\xi}_{k+1}) - \mat{D}( \greekvect{\xi}_{k})\mat{K}]\mat{a}_j\|_2^2 \eeq
where $\mat{D}( \greekvect{\xi}_k)$ is the matrix obtained by evaluating the dictionary a the point $\greekvect{\xi}_k$.}
 EDMD 
often works well, but requires choosing a suitable dictionary, with the approximation error strongly hinging on this choice. This approximation error can be reduced by considering larger dictionaries, but this may lead to overfitting of  the data and poor generalization capabilities
 \cite{Otto2019}.
 
 

 Deep learning motivated approaches use  a neural network  parameterized by a set of weights $W$ as dictionary.  The  \begin{wrapfigure}[23]{r}{0.5\textwidth}
 \includegraphics[width=0.5\textwidth]{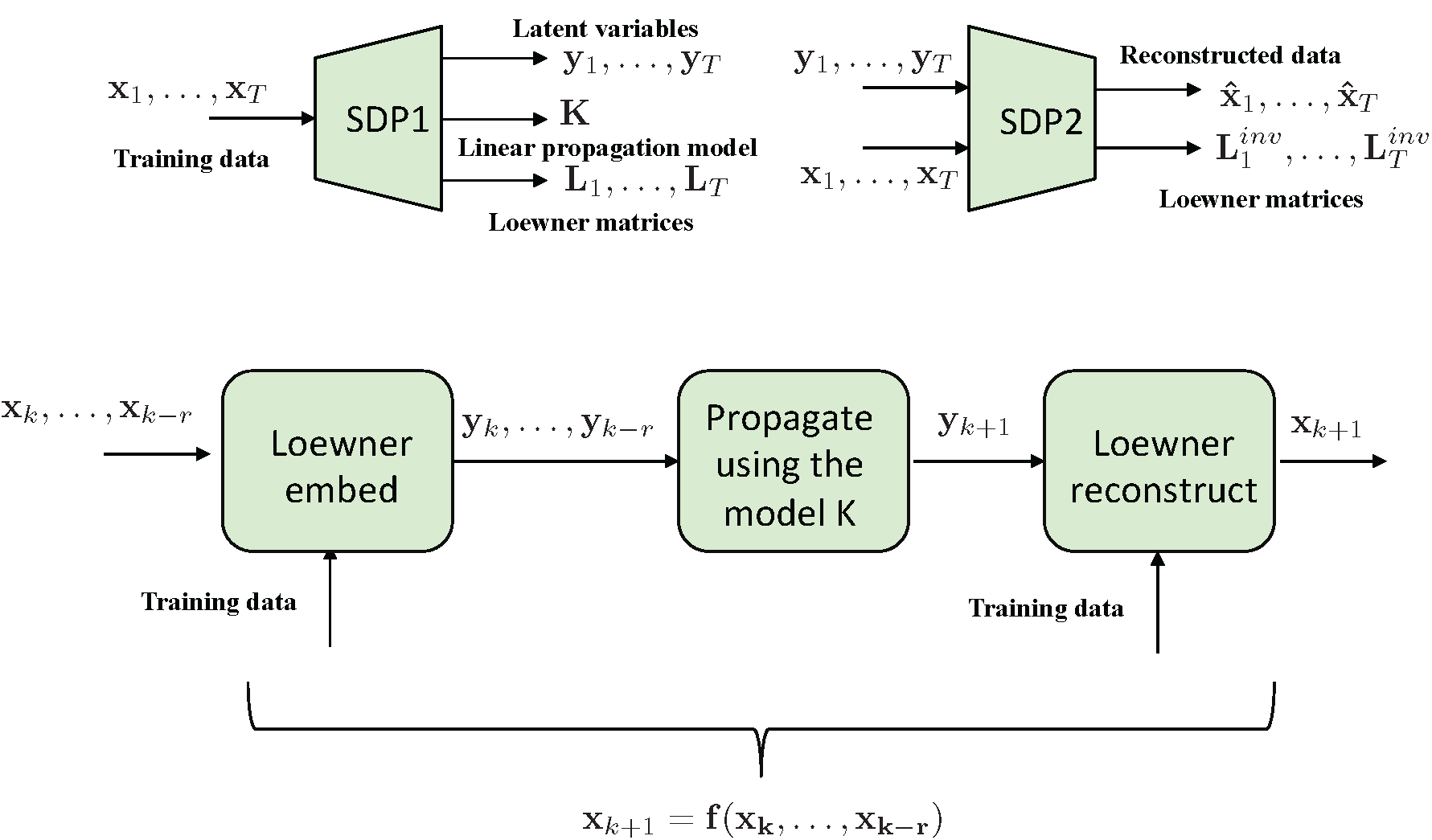}
 \caption{\small Top: Finding Koopman operators via Semi-Definite Programs. The first SDP (Section \ref{sec:relaxation}) finds the observables
 $\mat{y}_k$ corresponding to given data $\mat{x}_k$, the Koopman operator $\mat{K}$, and the Loewner matrices that encode the mapping $\mat{x}_k \to \mat{y}_k$. The second SDP (Section \ref{sec:mapping}) finds the inverse mapping $\mat{y}_k \to \mat{x}_k$. Bottom: The pipeline to predict $\mat{x}_{k+1}/
 \mat{x}_k \ldots, \mat{x}_{k-r}$ uses explicit expressions for  the predictions of the model $\mat{K}$. Thus, it only requires $\mathcal{O}(r)$ operations.}\label{fig:Koopman} \end{wrapfigure}  Koopman operator $\mat{K}$ is found by alternatively minimizing the  prediction error  over  $W$ and $\mat{K}$.  Alternating minimization methods  can get trapped in local minima. Further, the issue of which architectures are best suited to represent dynamical systems is largely open.  Recent work  \cite{Lusch2018,Otto2019}  proposed encoder/decoder type  architectures that map states $\greekvect{\xi}$ to latent variables $\mat{y}$ and impose approximately linear dynamics for the evolution of the latter. A salient feature of these approaches is that the states $\greekvect{\xi}$ are no longer required to be in the span of the Koopman eigenfuctions. As shown in \cite{Otto2019}, the use of a nonlinear decoder to map $\mat{y}$ back to   $\greekvect{\xi}$ (as opposed to a linear one if  $\greekvect{\xi} \in$ span \{ $\mathcal{D}$\}) results in substantially smaller dictionaries. Still, these methods require ad-hoc  parameter selection (dimension of the latent variables, order of the dynamics) and, as before, can lead to local minima.

 An alternative approach, HAVOK \cite{Havok}, rooted in Takens embedding theorem \cite{Takens}, seeks to model the trajectories of \eqref{eq:nonlinear} by considering a forced linear system, whose dynamics are precisely the Koopman operator. The  states and forcing term are obtained from the singular value decomposition of a Hankel matrix $\mat{H}_\mat{x}$, formed by delayed measurements of $\mat{x}_k$. As shown in \cite{Havok} this approach successfully recovers the trajectories of nonlinear chaotic systems, as linear combinations of a given basis.  However, this linear reconstruction, combined with the difficulty of  identifying the linear dynamics from the svd of $\mat{H}_\mat{x}$ \cite{Havok} can lead to high order models (e.g. a 14$^{th}$ order model for the third order Lorentz system).   
 
 In this paper, motivated by \cite{FeiICCV11, FeiHankel, Havok,Lusch2018,Otto2019}, we  propose an alternative, convex optimization based, approach to the problem of data-driven identification of Koopman operators.    The philosophy, illustrated in Fig. \ref{fig:Koopman}, uses delay coordinates, but, as in  \cite{Lusch2018,Otto2019}  does not impose that the state of the system belongs to span of the Koopman eigenfunctions. Rather, we identify a manifold of latent variables where the dynamics are linear and map back to state-space via a non-linear transformations.  The problems of finding the embedding manifold, the associated Koopman operators and the mapping back to state-space are all recast as  rank-constrained semi-definite programs (SDPs). In turn, these can be relaxed to convex optimizations using the standard weighted nuclear norm surrogate for rank. Advantages of the proposed approach include:
 \begin{itemize}
 \item  A simple rank check allows for certifying that the solution to these convex SDPs is indeed the Koopman operator underlying the given data.

  \item Does not specify a priory the dimension of the embedding or the order of the dynamics. Rather, both of these can be obtained from the solution to the SDPs.  
  \item Minimizing  the  order of the linear  dynamics leads to simpler models than competing methods. 
  \item In cases where the spectrum of the Koopman operator is not finite, it allows for obtaining finite dimensional approximations with guaranteed approximation error.
  
  \item These SDPs have an underlying structure, chordal sparsity, that can be exploited to substantially reduce computational complexity, leading to algorithms that scale linearly with the number of data points.

 
\end{itemize} 

 The paper is organized as follows. In section \ref{sec:preliminaries} we formally state the problem under consideration and summarize some needed results on  rational interpolation. Section \ref{sec:rank} contains the main results of the paper. It shows that a Hilbert space $\mathbb{H}$ of observables, its associated Koopman dictionary and eigenfunctions, and the mapping back to state-space can be found by solving rank-constrained SDPs.  Section  \ref{sec:examples}
 illustrates the proposed approach with some simple examples. Finally, Section \ref{sec:conclusions} summarizes the paper and points out to directions for extending its results.
 
\section{Preliminaries}\label{sec:preliminaries}
For ease of reference, next we summarize our notation and recall some results on interpolation.

\subsection{Notation} \label{sec:notation}

\begin{supertabular}{p{0.1\linewidth}p{0.9 \linewidth}}
$|\mathcal{S}|$ & cardinality of the set $\mathcal{S}$ \\
$\vect{x},\vect{M}$ & a vector in $\mathbb{R}^n$ (matrix in $\mathbb{R}^{n\times m}$) \\
$\otimes$ & Matrix Kronecker product \\
$\vect{M} \succeq 0$ & the matrix $\vect{M}$ is positive semidefinite. \\
$\|\mat{M}\|_*$ & nuclear norm: $\|\mat{M}\|_* = \Sigma \text{ singular values of $\mat{M}$}$.\\
$\mat{H}^{m}_y$ & Hankel matrix with $m$ columns associated with a vector sequence $\vect{y}(.)$, with block elements
$(\mat{H}_y^{m})_{i,j}=\mat{y}_{i+j-1}$ \\

\textbf{svec}$(\mat{M)}$ & (column-wise) vectorization of the unique elements of a symmetric matrix $\mat{M}$.\\

\textbf{smat}$(\mat{v)}$ & create a symmetric matrix $\mat{M}$ from the elements of $\mat{v}$ such that
\textbf{svec}$(\mat{M}) = \mat{v}$ \\
\end{supertabular}








\subsection{Rational Interpolants and Loewner Matrices} \label{sec:Loewner}

 Given $2n$ scalar pairs $(x_i,y_i)$, consider the problem of finding a rational function $g(x) \doteq \frac{\sum_{k=1}^m a_k x^k}{\sum_{k=1}^m b_k x^k}$ such that $y_i = g(x_i), \; i=1,\ldots 2n$. Define the Loewner matrix
 
$\mat{L} \in \mathbb{R}^{n\times n}$ as :
 \begin{equation}\label{eq:L2}
\mat{L}(x,y) = 
 \begin{bmatrix}
 \frac{y_1-y_{n+1}}{x_1-x_{n+1}}  & \frac{y_1-y_{n+2}}{x_1-x_{n+2}} & \ldots &\frac{y_1-y_{2n}}{x_1-x_{2n}} \\
 \frac{y_2-y_{n+1}}{x_2-x_{n+1}}  & \frac{y_2-y_{n+2}}{x_2-x_{n+2}} & \ldots &\frac{y_2-y_{2n}}{x_2-x_{2n}} \\
 \vdots & \vdots & \ddots & \vdots \\
 \frac{y_n-y_{n+1}}{x_n-x_{n+1}}  & \frac{y_n-y_{n+2}}{x_n-x_{n+2}} & \ldots &\frac{y_n-y_{2n}}{x_n-x_{2n}} 
  \end{bmatrix} 
 \end{equation}
 Then, there exists a  rational function of order at most $m$ that interpolates the given data points if and only if rank$(\mat{L}) \leq m-1$  \cite{Antoulas:1986,Ionita:2013}.
 
\subsection{Statement of the problem}\label{sec:statement}

Consider the nonlinear dynamical system:
\beq \label{eq:nonlinear1} \begin{split}
\mat{x}_{k+1} & = f(\mat{x}_k,\ldots,\mat{x}_{k-r+1})
 \; \qquad \mat{x}_j \in \mathbb{R}^n 
\end{split}
\eeq
where both the dynamics $f(.)$ and its order $r$ are unknown.  Our goal is to identify its associated Koopman operator, over a suitable  space of observables, from experimental data $\mat{x}$. Specifically:
\begin{Problem}\label{prob:problem1} Given a set of  $N$ trajectories $\{\mat{x}_k^{(i)}\}_{k=1}^{T_i},\;i=1,\ldots, N$, $\mat{x}_k^{(i)} \in \mathbb{R}^n$, find a (functional) dictionary $\mathcal{D}(.)$, a Hilbert space
$\mathbb{H}$ of observables $\greekvect{\psi}(.)$ of the form:
\beq \label{eq:psi}\begin{split}
&\greekvect{\psi}(\greekvect{\xi}_k) \doteq \begin{bmatrix}\mat{y}_{k-r+1}^T  \ldots  \mat{y}_{k}^T \end{bmatrix}^T \in \text{span $\{\mathcal{D}(\greekvect{\xi}_k) \}$ with 
 $\mat{y}_j \in \mathbb{R}^m$} \\
 &\text{where $\greekvect{\xi}_k \doteq \begin{bmatrix}\mat{x}_{k-r+1} \ldots \mat{x}_{k}\end{bmatrix}^T$}
\end{split}
\eeq
and an operator $\mathcal{K}\colon \mathbb{H} \to \mathbb{H}$ such that $(\mathcal{K}\circ \greekvect{\psi})(\greekvect{\xi}_k) =
 \greekvect{\psi}(\greekvect{\xi}_{k+1})$.
\end{Problem}


Problem \ref{prob:problem1}  is reminiscent of EDMD approaches. However, the main difference is that here  we seek to learn the dictionary $\mathcal{D}$  and the dimensions of the space $\mathbb{H}$  directly from the data, rather than postulating a fixed dictionary and dimension. Further, if Problem \ref{prob:problem1} has a solution, the resulting operator $\mathcal{K}$ is indeed the exact Koopman operator in $\mathbb{H}$.

\begin{Remark} As stated, Problem \ref{prob:problem1} is ill posed, since  $\|\greekvect{\psi}(.)\|$ can be arbitrarily small or large. To avoid this, and with an eye towards reconstruction of $\greekvect{\xi}$ from $\greekvect{\psi}$, we will impose the additional constraints: 
\[ \begin{split}
& \frac{1}{ M_{\ell}(\greekvect{\xi}_i)}\|\greekvect{\psi}(\greekvect{\xi}_i ) - \greekvect{\psi}(\greekvect{\xi}_j) \|_2 \leq  \|\greekvect{\xi}_i - \greekvect{\xi}_j\|_2 \leq  M_u(\greekvect{\xi}_i)\|\greekvect{\psi}(\greekvect{\xi}_i ) - \greekvect{\psi}(\greekvect{\xi}_j)\|_2  \;
 \text{$\forall  \greekvect{\xi}_j$ such that  $\|\greekvect{\xi}_i - \greekvect{\xi}_j\|_2 \leq \delta$} \\
&\mat{y}(\greekvect{\xi}_k)_j^T\mat{y}(\greekvect{\xi}_k)_j =  (\greekvect{\xi}_k)_j^T(\greekvect{\xi}_k)_j \;  j=1,\ldots r, \;\text{for all $k$ in a given set of ``anchor" points  $\mathbb{I}$}\end{split}
\]
where $\mat{y}_j$ denotes the
$j^{\rm{th}}$ block component of $ \greekvect{\psi}( \greekvect{\xi}_k)$ and
the scalar $\delta$ and the set of anchor points $\mathbb{I}$ are design hyperparameters. 
That is, we impose that (a) the mapping  $\phi \colon  \greekvect{\xi} \to \greekvect{\psi}$ and its inverse are locally Lipschitz continuous, with Lipschitz constants
$M_\ell(\greekvect{\xi}_i)$ and $M_u(\greekvect{\xi}_i)$; and (b)  the function $\greekvect{\psi}(.)$ is normalized to have components with unity gain at some given ``anchor" points.
\end{Remark}

\section{Learning Koopman Operators via  Semi Definite Optimization}\label{sec:rank}

In this section we present the main theoretical result of the paper: a reformulation of Problem \ref{prob:problem1} as a rank minimization subject to a positive semi-definite constraint. Since this problem is generically NP hard, we then develop a tractable convex relaxation, along with optimality certificates.

\subsection{Finding Koopman operators as a  constrained rank minimization} \label{sec:rank1}

Consider the following feasibility problem (in $\vect{y}, r, m$):
\begin{Problem}\label{prob:feas}
Given  a set of $N$ trajectories $\{\vect{x}_\ell^{(i)}\}_{\ell=1}^{T_i},\;i=1,\ldots N$, $\mat{x}_\ell^{(i)} \in \mathbb{R}^n$, find  scalars $r,m$ and $N$ trajectories $\vect{y}_k^{(i)} \in \mathbb{R}^m, k=1,\ldots,T_i$, such that the following holds:

\begin{align}
&\text{rank $(\vect{H_y}^{(r+1)}) \leq r $, where $\vect{H_y}^{(r+1)} \doteq \begin{bmatrix}  \vect{H}^{(r+1)}_{\vect{y}^{(1)} }\\ \vdots \\ \vect{H}^{(r+1)}_{\vect{y}^{(N)} } \end{bmatrix}$} \; \text{and} \;
  \vect{H}_{\vect{y}^{(i)}}^{(r+1)}\doteq\begin{bmatrix}
              \vect{y}_{1}^{(i)} & \vect{y}_{2}^{(i)} & \cdots & \vect{y}_{r+1}^{(i)} \\
                \vect{y}_{2}^{(i)} & \vect{y}_{3}^{(i)} & \cdots & \vect{y}_{r+2}^{(i)} \\
                \vdots  & \vdots  & \ddots & \vdots  \\
            \vect{y}_{T_i -r}^{(i)} & \vect{y}_{T_i -r+1}^{(i)} & \cdots & \vect{y}_{T_i}^{(i)}
            \end{bmatrix} \label{eq:feas1} \\
            &\begin{array}{l} \| \vect{x}_s- \vect{x}_t\|_2 \leq  M_u(\vect{x}_s)\|\vect{y}_s - \vect{y}_t\|_2    \\ 
             \|\vect{y}_s - \vect{y}_t\|_2   \leq   M_{\ell}(\vect{x}_s)\| \vect{x}_s - \vect{x}_t\|_2 \end{array} \Big \}
             \text{$\forall  (s,t)$ such that  $\|\vect{x}_s - \vect{x}_t\|_2 \leq \delta$}
              \label{eq:feas2} \\
 & \; \|\vect{y}_s\|_2 = \|\vect{x}_s\|_2 \; \text{for all $s \in \mathbb{I}$}   \label{eq:feas3}
\end{align}
\end{Problem}
As shown next, the solution to Problem \ref{prob:problem1} (e.g the dictionary $\mathcal{D}$, the embedding Hilbert space $\mathbb{H}$ and the associated Koopman operator) can be constructed from any feasible solution to \eqref{eq:feas1}-\eqref{eq:feas3}.

\begin{Theorem}\label{Teo:theo1} Let $(\mat{y}^{(i)}_k, r, m)$  denote a feasible solution to \eqref{eq:feas1}-\eqref{eq:feas3} with $\mat{y}_k^{(i)} \in \mathbb{R}^m$ and rank$( \vect{H_y}) =r* \leq r$. Let $\mat{H_y}^{(r^*+1)}$ and  $\mathcal{N}_R(\vect{H_y}^{(r^*+1)})$
denote the Hankel matrix obtained by rearranging the elements of $\mat{H_y}^{(r+1)}$ into
$r^*+1$ columns, and its right null space, respectively. Note that by construction rank$(\mat{H_y}^{(r^*+1)})=r^*$ and thus $dim(\mathcal{N}_R(\vect{H_y}^{(r^*+1)})) \geq 1$.
 Consider a vector $\mat{p} \in \mathcal{N}_R(\vect{H_y}^{(r^*+1}))$, of the form
$\mat{p}=\begin{bmatrix}a_0 & \ldots  &a_{(r^*-1)} & -1\end{bmatrix}^T$.  Let $\rho_j, \; j=1,\ldots,r^*$ denote the roots of the polynomial
$ \mathcal{P}(\rho) \doteq  \rho^{r^*} - \sum_{i=0}^{r^*-1} a_i \rho^{i} $ and define the $r^*$ vectors
\[\mat{v}_j \doteq \begin{bmatrix} 1 & \rho_j & \rho_j^2 & \ldots \rho_j^{r^*} \end{bmatrix}^T \]
Finally, let $\mat{V}$ denote the Vandermonde matrix $\mat{V}=\begin{bmatrix} \mat{v}_1 & \mat{v}_2 & \ldots & \mat{v}_{r^*} \end{bmatrix}$.
Then: \begin{enumerate}
\item The desired dictionary $\mathcal{D}(.)$  has the matrix representation $\mat{D} =  \mat{V}  \otimes \mat{I}_m$.
\item The Hilbert space $\mathbb{H}$ of observables is given by span($\mathcal{D}$), with the usual inner product.
\item The operator 
$\mathcal{K}\colon \mathbb{H}\to \mathbb{H}$ with the matrix representation $\greekvect{\Lambda}=\text{diag}(\rho_i)\otimes \mat{I}_m$ in the basis defined by the columns of  $\mat{V}$ is  the Koopman operator associated with \eqref{eq:nonlinear} in the space $\mathbb{H}$.

\end{enumerate}
\end{Theorem}
\begin{proof} Given in the Appendix.
  \end{proof}
 

Theorem \ref{Teo:theo1}  provides the  foundation for constructing the Koopman operator from the solution of an optimization problem, but is of limited practical value, due to several reasons: (i) It does not indicate how to find  $m$, the dimension of $\mat{y}_k$, or $r$, the ``memory" of the system, and (ii) it leads to a difficult, non-convex problem.  Motivated by \cite{FeiICCV11}, next we show that Problem \ref{prob:feas} is equivalent to a SDP constrained rank-minimization.   The starting point is to consider the Kernel matrix  with entries $\mat{K}_{r,s}\doteq \mat{y}_{r}^T\mat{y}_{s}$, where $\mat{y}_{r}, \mat{y}_{s}$ denote the
observables corresponding to points $\mat{x}_{r},\mat{x}_s$ drawn from (not necessarily the same) training trajectories. Let $\mat{y}_{s}^{(i)}, \; s=1,\ldots T_i$ denote the observables corresponding to the i$^{\rm{th}}$ trajectory and define  the  $(r+1)\times(r+1)$ Gram matrix  $\mat{G}^{(i)} \doteq 
 (\mat{H}^{(r+1)}_{\mat{y}^{(i)}})^T \mat{H}^{(r+1)}_{\mat{y}^{(i)}}$.   The key observation is that both the entries of $\mat{G}^{(i)}$ and the argument of the constraints \eqref{eq:feas2}--\eqref{eq:feas3} are \emph{affine} functions of entries of $\mat{K}$, leading to the following result:

\begin{Theorem}\label{Theo:SDP}   Define the family of Gram matrices:
 $\vect{G}^{(i)}=  (\mat{H}^{(r+1)}_{\mat{y}^{(i)}})^T \mat{H}^{(r+1)}_{\mat{y}^{(i)}}=\sum_{\ell=0}^{T_i-r}\vect{K}_{\ell,r}^{(i)}$ where
\[ \vect{K}_{\ell,r}^{(i)}=\begin{bmatrix}
              (\mat{y}_{\ell}^{(i)})^T\mat{y}_{\ell }^{(i)} & (\mat{y}_{\ell}^{(i)})^T\mat{y}_{\ell +1}^{(i)} & \cdots &  (\mat{y}_{\ell}^{(i)})^T\mat{y}_{\ell +r}^{(i)} \\
                \vdots  & \vdots  & \ddots & \vdots  \\
          (\mat{y}_{\ell+r}^{(i)})^T\mat{y}_{\ell}^{(i)}  &  (\mat{y}_{\ell+r}^{(i)})^T\mat{y}_{\ell+1}^{(i)} & \cdots &  (\mat{y}_{\ell+r}^{(i)})^T\mat{y}_{\ell+r}^{(i)}\\
            \end{bmatrix}
            \]
(note that $\vect{K}_{\ell,r}^{(i)}$ are submatrices of $\mat{K}$). Consider the following rank minimization problem:
\begin{align}
r^*=  &  \min_{\mat{K} \succeq 0} \text{rank}(\vect{G} \doteq \begin{bmatrix} (\vect{G}^{(1)})^T & \ldots & 
 (\vect{G}^{(N)})^T \end{bmatrix}^T )  \;  \text{subject to:} \label{eq:Kc0} \\
  & \left .\begin{array}{l} \frac{1}{ M^2_u(\vect{x}_s)} \| \vect{x}_s - \vect{x}_t\|_2^2 \leq  K_{s,s}-2K_{s,t}+K_{t,t} \\
 K_{s,s}-2K_{s,t}+K_{t,t} \leq M^2_{\ell} (\vect{x}_s)  \| \vect{x}_s - \vect{x}_t\|^2_2  \end{array}  \right \}
 \text{$\forall  (s,t)$ such that  $\|\vect{x}_s - \vect{x}_t\|_2 \leq \delta$} \label{eq:Kc1} \\
 & K_{s,s} = \|\vect{x}_s\|_2^2 \; \text{for all $s \in \mathbb{I}$}  \label{eq:Kc2}
    \end{align}
    Denote by $\mat{K}^{(i)}$ the submatrix of $\mat{K}$ with entries $(\mat{K}^{(i)})_{\ell,j}=(\mat{y}^{(i)}_{\ell})^T\mat{y}^{(i)}_{j}$, and
    let $m= \max_i\left \{\text{rank}(\mat{K}^{(i)})\right \}$. Consider the factorizations $(\mat{Y}^{(i)})^T\mat{Y}^{(i)} = \mat{K}^{(i)}$ with $\mat{Y}^{(i)} \in \mathbb{R}^{m\times T_i}$.
    Then, if  $r* < r+1$, the columns $\mat{y}_k^{(i)}$ of $\mat{Y}^{(i)}$ solve Problem \ref{prob:feas}.  \end{Theorem}
\begin{proof} Given in the Appendix \end{proof}
\subsection{Adding a regularization}\label{sec:regularized}
Theorems \ref{Teo:theo1}  indicates how to find  the observables $\greekvect{\psi}(.) \in \mathbb{H}$ by solving a constrained optimization problem. Further, these constraints guarantee that the mapping $\greekvect{\psi}(.)\colon \mathbb{R}^{rn}  \to \mathbb{H}$ locally satisfies some Lipschiz and gain constraints. However these constraints alone do not guarantee that  $\greekvect{\psi}(.)$ is not arbitrarily complex, or even has the same functional form for all $\greekvect{\xi}$. These issues  can complicate the task of finding an explicit form for the mapping, if one is needed.  Next, we briefly indicate how to use additional degrees of freedom available in the problem to guarantee that $\greekvect{\psi}(.)$ is the simplest possible mapping, in a sense precisely defined below, and has the same functional form for all  $\greekvect{\xi}$.

Consider a point $\greekvect{\xi}_k \doteq  \begin{bmatrix} \mat{x}_{k-r+1}^T &\ldots & \mat{x}_{k}^T\end{bmatrix}^T$, and for
each (block) component $\mat{x}_j$, denote by $\mathbb{N}_{\mat{x}_j}$ the indexes of its nearest neighbors.  Let $\mat{K}_{\mat{x}_j}$ be matrix
with elements $(\mat{K}_{\mat{x}_j})_{r,s} = \mat{x}^T_r\mat{x}_s$ for all $r,s \in \left \{ j\cup \mathbb{N}_{\mat{x}_j} \right \}$. Similarly, given 
$\greekvect{\psi}(\greekvect{\xi}_k) \doteq \begin{bmatrix}\mat{y}_{k-r+1} \ldots \mat{y}_{k}\end{bmatrix}^T$, let $\mat{K}_{\mat{y}_j}$ be the submatrix of $\mat{K}$  with elements $(\mat{K}_{\mat{y}_j})_{r,s} = \mat{y}^T_r\mat{y}_s$ for all $r,s \in \left \{ j\cup \mathbb{N}_{\mat{x}_j} \right \}$.  
 For ease of notation, let
  $\greekvect{\kappa}_{\mat{x}}^{(j)} \doteq  \textbf{svec$(\mat{K}_{\mat{x}_j}) $} \in  \mathbb{R}^q $, 
   $\greekvect{\kappa}_{\mat{y}}^{(j)} \doteq \textbf{svec$(\mat{K}_{\mat{y}_j})$} \in  \mathbb{R}^q$, where  
   $q \doteq \frac{(|\mathbb{N}_{\mat{x}_j}|+1 )(|\mathbb{N}_{\mat{x}_j}| +2)}{2}$.  Note that these vectors contain the unique elements of the matrices 
   $\mat{K}_{\mat{x}_j}$, $\mat{K}_{\mat{x}_j}$. Finally, let $p = \lfloor \frac{q}{2} \rfloor$ and 
define the Loewner matrix
\beq \label{eq:Loewner1}
\mat{L}_{\mat{x}_j} \doteq 
 \begin{bmatrix}
 \frac{\greekvect{\kappa}_{\mat{y}_1}-\greekvect{\kappa}_{\mat{y}_{p+1}}}{\greekvect{\kappa}_{\mat{x}_1}-\greekvect{\kappa}_{\mat{x}_{p+1}}}&
  \frac{\greekvect{\kappa}_{\mat{x}_1}-\greekvect{\kappa}_{\mat{x}_{p+2}}}{\greekvect{\kappa}_{\mat{y}_1}-\greekvect{\kappa}_{\mat{x}_{p+2}}}
 & \ldots & \frac{\greekvect{\kappa}_{\mat{y}_1}-\greekvect{\kappa}_{\mat{y}_{q}}}{\greekvect{\kappa}_{\mat{x}_1}-\greekvect{\kappa}_{\mat{x}_{q}}} \\
 \vdots & \vdots & \ddots & \vdots \\
  \frac{\greekvect{\kappa}_{\mat{y}_{p}}-\greekvect{\kappa}_{\mat{y}_{p+1}}}{\greekvect{\kappa}_{\mat{x}_{p}}-\greekvect{\kappa}_{\mat{x}_{p+1}}}&
 & \ldots & \frac{\greekvect{\kappa}_{\mat{y}_{p}}-\greekvect{\kappa}_{\mat{y}_{q}}}{\greekvect{\kappa}_{\mat{x}_{p}}-\greekvect{\kappa}_{\mat{x}_{q}}} \\
  \end{bmatrix}   \eeq
  where $\greekvect{\kappa}_{\mat{x}_i}, \greekvect{\kappa}_{\mat{y}_i}$ denote the $i^{\rm{th}}$ component of $\greekvect{\kappa}_{\mat{x}}^{(j)}$ and   $\greekvect{\kappa}_{\mat{y}}^{(j)}$ respectively. From the results in section \ref{sec:Loewner}, it follow that if rank$(\mat{L}_{\mat{x}_j}) < p$, then there exists a rational mapping of degree up to $p-1$ that maps the elements of $\mat{K}_{\mat{x}_j}$ to those of $\mat{K}_{\mat{y}_j}$. Further, the degree of this mapping can be minimized by minimizing the rank  of $\mat{L}_{\mat{x}_j}$ with respect to the variables $\greekvect{\kappa}_{\mat{y}_i}$, leading (locally) to the lowest order rational mapping $\greekvect{\psi}(\greekvect{\xi}_k)$. If a global, rather than local, rational mapping is desired, a similar idea can be using involving all pairs $\mat{x},\mat{y}$, rather than just the nearest neighbors of each point.  

\subsection{A Convex Relaxation}\label{sec:relaxation}

Theorem \ref{Theo:SDP} allows for reducing Problem \ref{prob:problem1} to a constrained rank minimization problem. However, this problem is still NP-hard. In order to obtain a tractable relaxation,  we will replace the objective \eqref{eq:Kc0} by  $\sum_{i=1}^N \text{rank}(\vect{G}^{(i)})$ and add a term of the form $\lambda_1 \sum_{j=1}^T \text{rank}(\vect{L}_{\mat{x}_j})$, where $T=\sum T_i$ is the total number of points.
Then, proceeding as in \cite{MohanFazel},  we will replace rank with a convex surrogate, a weighted nuclear norm, where the weights are updates as each step of the algorithm. Finally, in order to handle  outliers, we will consider a ``soft" version of  \eqref{eq:Kc1}-\eqref{eq:Kc2}, where these are added to the objective as penalties. The complete algorithm is outlined in Algorithm \ref{alg:algorithm2}. It is worth noting that if the algorithm yields a solution $\mat{G}$ with rank$(\mat{G}) < r$, this certifies that $m$ is indeed the Koopman operator. On the other hand, if the algorithm yields a solution $\mat{G}$ with minimum singular value $\sigma_{\text{min}}$, then an $r^{th}$ order approximate model $m_r$ can be obtained by performing PCA on
$\mat{G}$. In this case the approximation error is bounded (in the Hankel norm sense) by $\sqrt{\sigma_{\text{min}}}$.

\begin{algorithm}[hbt]
    \caption{Reweighted $\|.\|_*$ based Koopman Identification} \label{alg:algorithm2}
    \begin{algorithmic}[1]
        \State \textbf{initialize:} \begin{tabular}{l}$iter = 0,  \mat{W}_0 = \mat{I};  \mat{V}_0^{(j)} = \mat{I}, j=1,\ldots,T$;        $\lambda_1,\lambda_2, \lambda_3, \delta \gets$ hyperparameters, \\ $\mathbb{NN}=\left \{ (s,t) \colon \|\vect{x}_s - \vect{x}_t\|_2 \leq \delta \right \}$,  $\sigma \gets $small number, $r \gets$ upper bound on system order.

       \end{tabular}
        \State \textbf{Repeat:}  Solve
        \[ \begin{array}{l}
        \min_{\mat{K}^{(i)} \succeq 0}\quad  \|\mat{W}_{iter}\mat{G}\|_* + \lambda_1  \sum_{j=1}^T \|\mat{V}^{(j)}_{iter}\mat{L}_{\mat{x}_j}\|_* + 
        \lambda_2  \sum_{s \in \mathbb{I}} (K_{s,s} - \|\vect{x}_s\|_2^2)^2 \\
        + \lambda_3 \sum_{{r,t} \in \mathbb{NN}} \max \left \{0,\frac{1}{ M^2_u(\vect{x}_s)} \| \vect{x}_s - \vect{x}_t\|^2_2 - K_{s,s}+2K_{s,t}-K_{t,t} \right \} \\
        + \lambda_3 \sum_{{r,t} \in \mathbb{NN}}  \max \left \{0, K_{s,s}-2K_{s,t}+K_{t,t} - M^2_{\ell} (\vect{x}_s)  \| \vect{x}_s - \vect{x}_t\|^2_2  
\right \} 
     \end{array} \]
                      Update
        \[ \begin{array}{l}
        \mat{W}_{(iter+1)} = \left ( \frac{ \mat{G} + \sigma\mat{I}}{\|\mat{G} + \sigma\mat{I}\|} \right )^{-1}, \;
        \mat{V}^{(j)}_{(iter+1)} = \left ( \frac{ \mat{L}_{\mat{x}_j} + \sigma\mat{I}}{\|\mat{L}_{\mat{x}_j} + \sigma\mat{I}\|} \right )^{-1}, \;
        iter = iter + 1
        \end{array}\]
        \State \textbf{Until:} rank$(\mat{G}) < r$.
        \State $[\mat{U}^{(i)},\mat{S}^{(i)},(\mat{U}^{(i)})^T] \gets \text{svd}(\mat{K}^{(i)})$, $\mat{S}^{(i)}\gets \frac{\mat{S}^{(i)}}{\|\mat{S}^{(i)}\|}$, $r_k^{(i)}\gets \min r \colon \sum_{j=1}^r \mat{S}_{jj}^{(i)} \geq 0.99$
        \State $\mat{Y}^{(i)} \gets [\mat{U}^{(i)}(:,1:r_k)]^T$ 
 \State       $[\mat{U}_{\mat{G}},\mat{R},\mat{V}_\mat{G}^T] \gets \text{svd}(\mat{G})$,  $\mat{m} \gets \mat{V}_\mat{G}(:,r+1)$ 
 \State \textbf{Output:} embeddings $\mat{Y^{(i)}}$, model $\mat{m}$.
   \end{algorithmic}
\end{algorithm}

\subsection{Mapping observables to states} \label{sec:mapping}

The approach presented in Section \ref{sec:rank} finds the observables $\greekvect{\psi}(\greekvect{\xi}_k)$ corresponding to a given trajectory $\greekvect{\xi}_k, \; k=1,\ldots T$. However, it does not explicitly provide a method for mapping  a  given $\greekvect{\psi}(\greekvect{\xi})$, obtained for instance by using the Koopman operator to propagate a trajectory in observable space, back to the corresponding point  $\greekvect{\xi}$ in state space. Motivated by \cite{SaulLLE} we propose to find (pointwise) the mapping
$\greekvect{\psi} \to \greekvect{\xi}$ by locally approximating the mapping between the embedded space and ambient space kernels,  $\mat{K}_\mat{y}$ and $\mat{K}_\mat{x}$, with a rational function.  Specifically, given a point $\mat{y}^* \in \mathbb{R}^m$, 
 let $\mathcal{N}_{\mat{y^*}} \doteq \left \{ \mat{y}_k \colon \|\mat{y}^*-\mat{y}_k\|_2^2 \leq \delta \right \}$ and denote by $\mathcal{X}$ its preimage.  We propose to estimate $\mat{x}^*$ by first finding ${K}_{\mat{x}_i,\mat{x}^*}$, the elements of $\mat{K}_x$  corresponding to $\mat{x}_i^T\mat{x}^*, \; \forall \mat{x}_i \in \mathcal{X}$  and then finding $\mat{x}^*$ by factorizing  $\mat{K}_x$.  Note that, in order to get a valid kernel compatible with the priors, the elements  ${K}_{\mat{x}_i,\mat{x}^*}$ should be such that the completed matrix $\mat{K}_{\mat{x}} \succeq 0$, rank$(\mat{K}) \leq n$, and the constraints  \eqref{eq:feas2}-\eqref{eq:feas3} are satisfied.  As shown next, under the assumption that the mapping $G\colon \mat{K}_\mat{y} \to \mat{K}_\mat{x}$ is rational, then $\mat{x}^*$ can be found by solving a rank minimization problem subject to semi-definite constraints.
 
 Consider the  Kernel matrices $\mat{K}_{\mat{x}}, \mat{K}_{\mat{y}} \in \mathbb{R}^{(|\mathcal{X}|+1 )\times (|\mathcal{X}|+1 )}$, where the entries have been ordered so that the elements of the form  $\mat{y}_i^T\mat{y}^*$ and $\mat{x}_i^T\mat{x}^*$ appear in the first row and column. As before, for ease of notation, let
  $\greekvect{\kappa}_\mat{x} = \textbf{svec$(\mat{K}_\mat{x}) $}$, 
   $\greekvect{\kappa}_\mat{y} = \textbf{svec$(\mat{K}_\mat{y})$}$.  Note that $\greekvect{\kappa}_\mat{x},\greekvect{\kappa}_\mat{y} \in \mathbb{R}^q$, with 
   $q \doteq \frac{(|\mathcal{X}|+1 )(|\mathcal{X}| +2)}{2}$, and that all inner products involving $\mat{y}^*$ and $\mat{x}^*$ appear in the first $|\mathcal{X}|+1$ elements of  $\greekvect{\kappa}_\mat{y}$ and
    $\greekvect{\kappa}_\mat{x}$. Let $p=\lfloor\frac{q}{2\rfloor}$ and consider 
the following rank minimization problem:
 
\begin{align}  \label{eq:SDPx}
&\min_{\greekvect{\kappa}_\mat{x}} \text{rank $(\mat{L})$}  \text{subject to:} \\
 & \left .\begin{tabular}{l}
$\kappa_{\mat{x}_1} -2 \greekvect{\kappa}_{\mat{x}_i}+ \greekvect{\kappa}_{\mat{x}_j}
  \leq \delta^2$ 
  \\
$\greekvect{\kappa_{\mat{x}_1}} -2 \greekvect{\kappa}_{\mat{x}_i}+ \greekvect{\kappa}_{\mat{x}_j}
  \leq M_u^2 (\vect{x}_i)( \kappa_{\mat{y}_1} -2 \greekvect{\kappa}_{\mat{y}_i}+ \greekvect{\kappa}_{\mat{y}_j})$
   \\
 $\greekvect{\kappa_{\mat{y}_1}} -2 \greekvect{\kappa}_{\mat{y}_i}+ \greekvect{\kappa}_{\mat{y}_j}
  \leq M_{\ell}^2 (\vect{x}_i)( \kappa_{\mat{x}_1} -2 \greekvect{\kappa}_{\mat{x}_i}+ \greekvect{\kappa}_{\mat{x}_j})$
  \\
  \end{tabular} \right \}  \begin{tabular}{l}
  $i=2,\ldots,|\mathcal{X}|+1$ \\
  $ j=\frac{(2|\mathcal{X}| +4-i)(i-1)}{2}+1$
  \end{tabular} 
   \label{eq:SDPx4}
\end{align}
\begin{align}
& \mat{K}_{\mat{x}} \doteq \textbf{smat}(\greekvect{\kappa}_{\mat{x}}) \succeq 0, \;  
 \text{rank}(\mat{K}_{\mat{x}}) \leq n \label{eq:SDP0},  \greekvect{\kappa}_{\mat{x}_1}= \greekvect{\kappa}_{\mat{y}_1}  \\
& \mat{L} = 
 \begin{bmatrix}
 \frac{\greekvect{\kappa}_{\mat{x}_1}-\greekvect{\kappa}_{\mat{x}_{p+1}}}{\greekvect{\kappa}_{\mat{y}_1}-\greekvect{\kappa}_{\mat{y}_{p+1}}}&
  \frac{\greekvect{\kappa}_{\mat{x}_1}-\greekvect{\kappa}_{\mat{x}_{p+2}}}{\greekvect{\kappa}_{\mat{y}_1}-\greekvect{\kappa}_{\mat{y}_{p+2}}}
 & \ldots & \frac{\greekvect{\kappa}_{\mat{x}_1}-\greekvect{\kappa}_{\mat{x}_{q}}}{\greekvect{\kappa}_{\mat{y}_1}-\greekvect{\kappa}_{\mat{y}_{q}}} \\
 \vdots & \vdots & \ddots & \vdots \\
  \frac{\greekvect{\kappa}_{\mat{x}_{p}}-\greekvect{\kappa}_{\mat{x}_{p+1}}}{\greekvect{\kappa}_{\mat{y}_{p}}-\greekvect{\kappa}_{\mat{y}_{p+1}}}&
 & \ldots & \frac{\greekvect{\kappa}_{\mat{x}_{p}}-\greekvect{\kappa}_{\mat{x}_{q}}}{\greekvect{\kappa}_{\mat{y}_{p}}-\greekvect{\kappa}_{\mat{y}_{q}}} \\
  \end{bmatrix}   \nonumber
  \end{align}
 
 \begin{Theorem}\label{teo:lowner} Let $\greekvect{\kappa}^*_\mat{x}, \mat{L}^*$ denote the solution to \eqref{eq:SDPx}-\eqref{eq:SDP0}. If rank$(\mat{L}^*) < p$, then
 (i) there exist a rational function $g(.)$ of degree at most $p$ such that $g(\kappa_{\mat{y}_i})=\kappa_{\mat{x}_i}$;  and  (ii) the vector $\mat{x}^*$ defined by the first row of $\mat{X}$, where $ \mat{X}^T\mat{X}=\mat{K}_{\mat{x}}$ satisfies constraints \eqref{eq:feas2}-\eqref{eq:feas3} in Problem \ref{prob:feas}.
 \end{Theorem}
 \begin{proof} Given in the Appendix \end{proof}
 
 Relaxing the rank in \eqref{eq:SDPx}  and \eqref{eq:SDP0} to a weighed nuclear norm, leads to an algorithm  similar  to Algorithm \ref{alg:algorithm2}, based on solving a sequence of SDPs until  rank deficient matrices $\mat{L}, \mat{K_x}$ are obtained.

 \section{Illustrative Examples} \label{sec:examples}
 
\noindent \textbf{Example 1: Lorentz Attractor.} In this example we consider the  Lorentz chaotic system:
\beq \label{eq:Lorentz} \begin{aligned}
 \dot{x}_1  =  \sigma(x_2 - x_1);
 \dot{x}_2  = x_1(\rho - x_3) - x_2;
  \dot{x}_3  = x_1x_2 - \beta x_3;
 \end{aligned}
 \eeq
with parameters $\sigma=28, \rho=10,\beta=\frac{8}{3}$. We used 400 points of the trajectory  starting at   $[-10.38 \; -4.5366 \; 35.1640]^T$,  uniformly sampled every $0.0271$ seconds to find the embeddings, and matlab's command {\tt ssest} to estimate an 7$^{\rm{th}}$ order model. Fig \ref{fig:Lorentz1}(a) shows the training and one step ahead reconstructed data, that is the results of  applying the encoder/decoder illustrated on the top of Fig. \ref{fig:Koopman} to (i) train, (ii) project the training data,  (iii) perform a one step ahead prediction and  (iv) lift back.  Figure \ref{fig:Lorentz1}(b) shows the predictions obtained using the pipeline at the bottom of Fig. \ref{fig:Koopman}, for points not part of the training data. As shown there, the proposed pipeline is indeed able to predict with reasonable accuracy the one step ahead value of the trajectory, using a 7$^{th}$ order Koopman operator. For comparison, \cite{Havok} uses a 14$^{th}$ order model.
 
 \begin{figure}[h]
 \begin{center}
 \includegraphics[width=2.3 in]{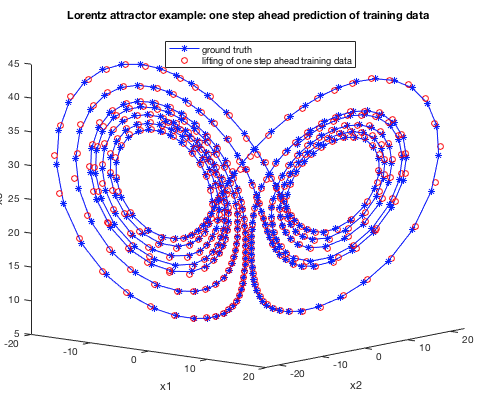}
\includegraphics[width=2.3 in]{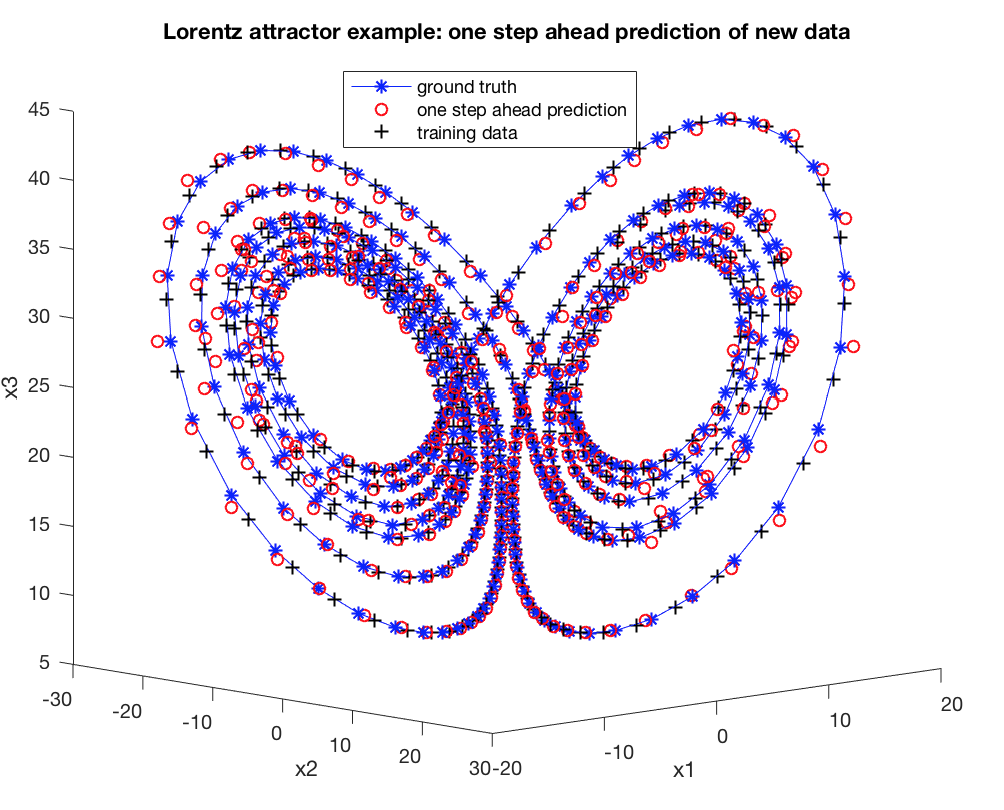}
\caption{\small Lorentz attractor: Left: one step ahead prediction of training data. Right: one step ahead prediction of new data}
\label{fig:Lorentz1}
\end{center}
\end{figure}

 \noindent \textbf{Example 2: The Duffing Oscillator.} Here we consider the system\footnote{The conventional Duffing equation is a forced oscillator. Here we use the last two equations to generate the forcing term $sin(t)$.}: 
 \beq \label{eq:duffing} \begin{aligned}
 \dot{x}_1 & = x_2;   \dot{x}_2 & = -0.5x_2 -x_1 - x_1^3 + 0.42 x_3; \;
 \dot{x}_3 &= x_4; \;
 \dot{x}_4 & =-x_3 
 \end{aligned}
 \eeq

 In this case,  Algorithm 1 yielded an embedding $\mat{y} \in \mathbb{R}^3$. We then used matlab's command {\tt ssest} to estimate a second order model for each component of $\mat{y}$.   Fig \ref{fig:duffing} (left) shows the one step ahead prediction of the training  data.  The right panel in Fig. \ref{fig:duffing} shows the predictions obtained using the pipeline at the bottom of Fig. \ref{fig:Koopman}, for points not part of the training data. As before, the proposed pipeline successfully   predicts the next point in the trajectory. 
 \begin{figure}[h]
 \begin{center}
 \includegraphics[width=2.0 in]{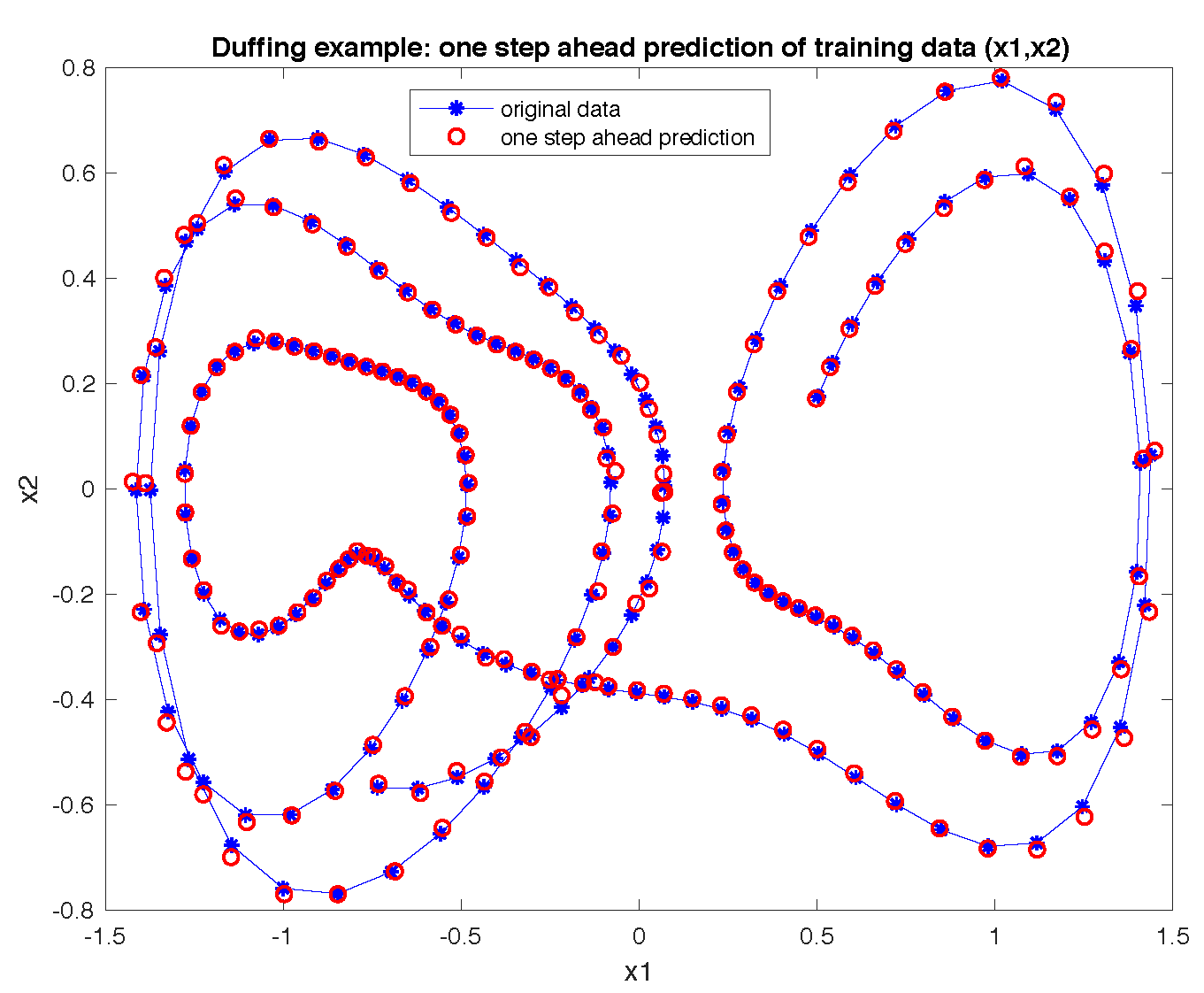}
\includegraphics[width=2.0 in]{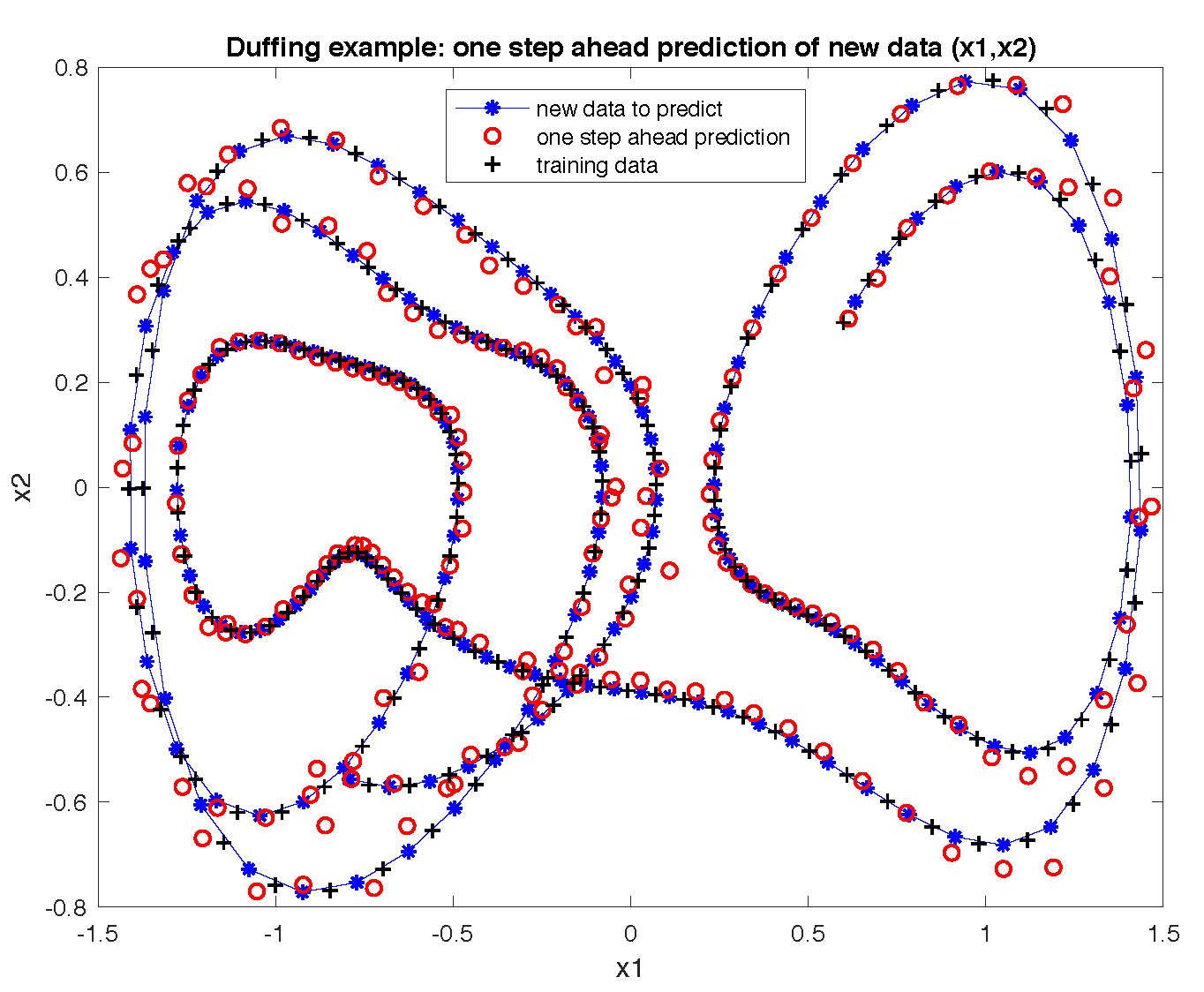}

\caption{\small Duffing oscillator one step ahead predictions of (left) training data and (right) new data.}\label{fig:duffing}
\end{center}
\end{figure}

\noindent \textbf{Example 3: predator-prey model.} In this example we considered the predator-prey model:\beq \label{eq:pp} \begin{aligned}
 \dot{x}_1 & = - x_1 + x_1x_2 \\
 \dot{x}_2 & = x_2 -x_1x_2
 \end{aligned}
 \eeq
 We used 120 points from the trajectory in ambient space to find the embeddings, and matlab's command {\tt ssest} to estimate an 8$^{\rm{th}}$ order model. Fig \ref{fig:pp}(a) shows the training and reconstructed data, that is the results of applying back to back the encoder/decoder illustrated on the top of Fig. \ref{fig:Koopman}.  Figure \ref{fig:pp}(b) shows the predictions obtained using the pipeline at the bottom of Fig. \ref{fig:Koopman}, starting from an initial condition not part of the training data. As shown there, the proposed pipeline is indeed able to predict with reasonable accuracy the trajectory over an 80 steps horizon that encompasses all regions visited by the trajectory.
 \vskip -1em
\begin{figure}[h]
\begin{center}
\begin{tabular}{cc}
\includegraphics[width=2 in]{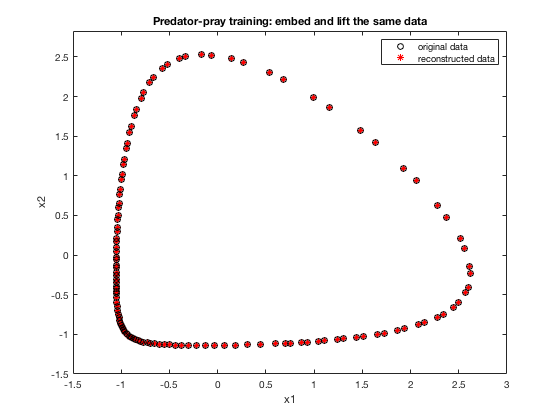}& \includegraphics[width=2 in]{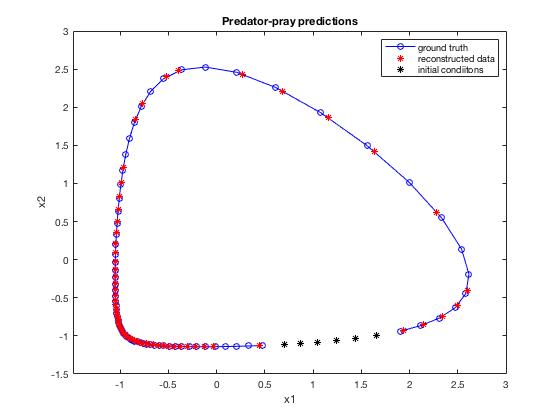}
\end{tabular}\caption{\small Predator-prey example. Left: encoding (black) and decoding (red) the training data. Right: one step ahead state prediction over an 80 step horizon, blue: ground truth, red: predictions using the proposed pipeline, black: initial conditions.}
\label{fig:pp}
\end{center}
\end{figure}

 \begin{wrapfigure}[11]{r}{2.1in}
 \vskip -2em
 \includegraphics[width=2.1in]{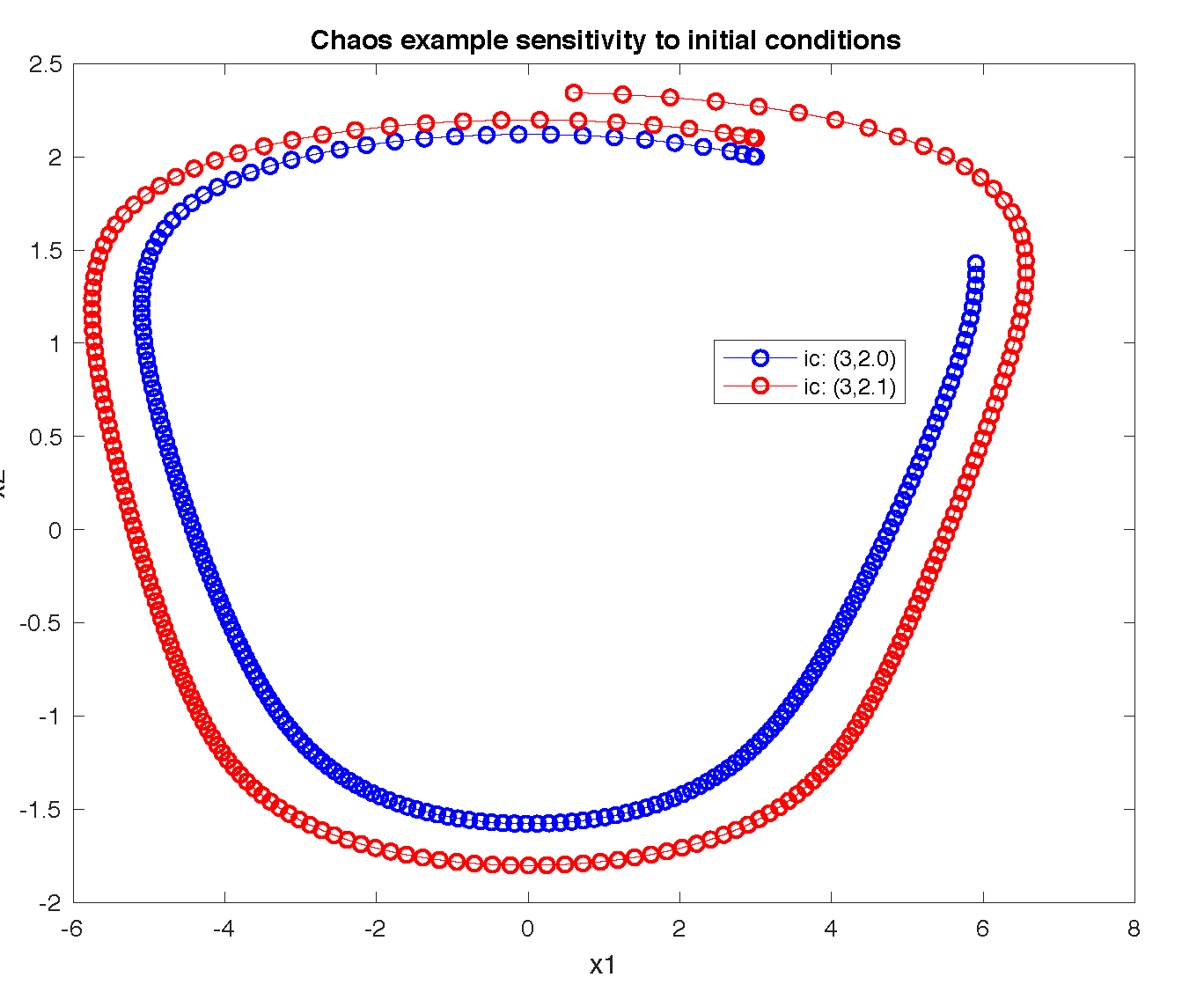}\caption{\small Two different trajectories corresponding to close initial conditions.}
\label{fig:chaos1}
\end{wrapfigure}
 \noindent \textbf{Example 4: Another chaotic system.} Here we consider the system: 
 \beq \label{eq:chaos} \begin{aligned}
 \dot{x}_1 & = x_2 \\
 \dot{x}_2 & = -x_1^5 -0.1x_2 +x_3 \\
 \dot{x}_3 &= x_4  \\
 \dot{x}_4 & =-x_3 
 \end{aligned}
 \eeq
 It is well know that this system has extreme sensitivity to initial condition. This effect is illustrated in Fig \ref{fig:chaos1} showing two different trajectories corresponding to the initial conditions $[2.0,3.0,0.0,6.0]$ and $[2.1,3.0,0.0,6.0]$.  We used 100 points from the trajectory in ambient space to find the embeddings. In this case Algorithm 1 yielded an embedding $\mat{y} \in \mathbb{R}^3$. We then used matlab's command {\tt ssest} to estimate a second order model for each component of $\mat{y}$.   Fig \ref{fig:chaos2}(left) shows the training and reconstructed data.  The right panel in Figure \ref{fig:chaos2} shows the predictions obtained using the pipeline at the bottom of Fig. \ref{fig:Koopman}, starting from an initial condition not part of the training data. As before, the proposed pipeline successfully   predicts with reasonable accuracy the trajectory over a 100 steps horizon, in spite of the sensitivity of the system to initial conditions noted above.
 
 \begin{figure}[h]
 \begin{center}
 \includegraphics[width=2.1 in]{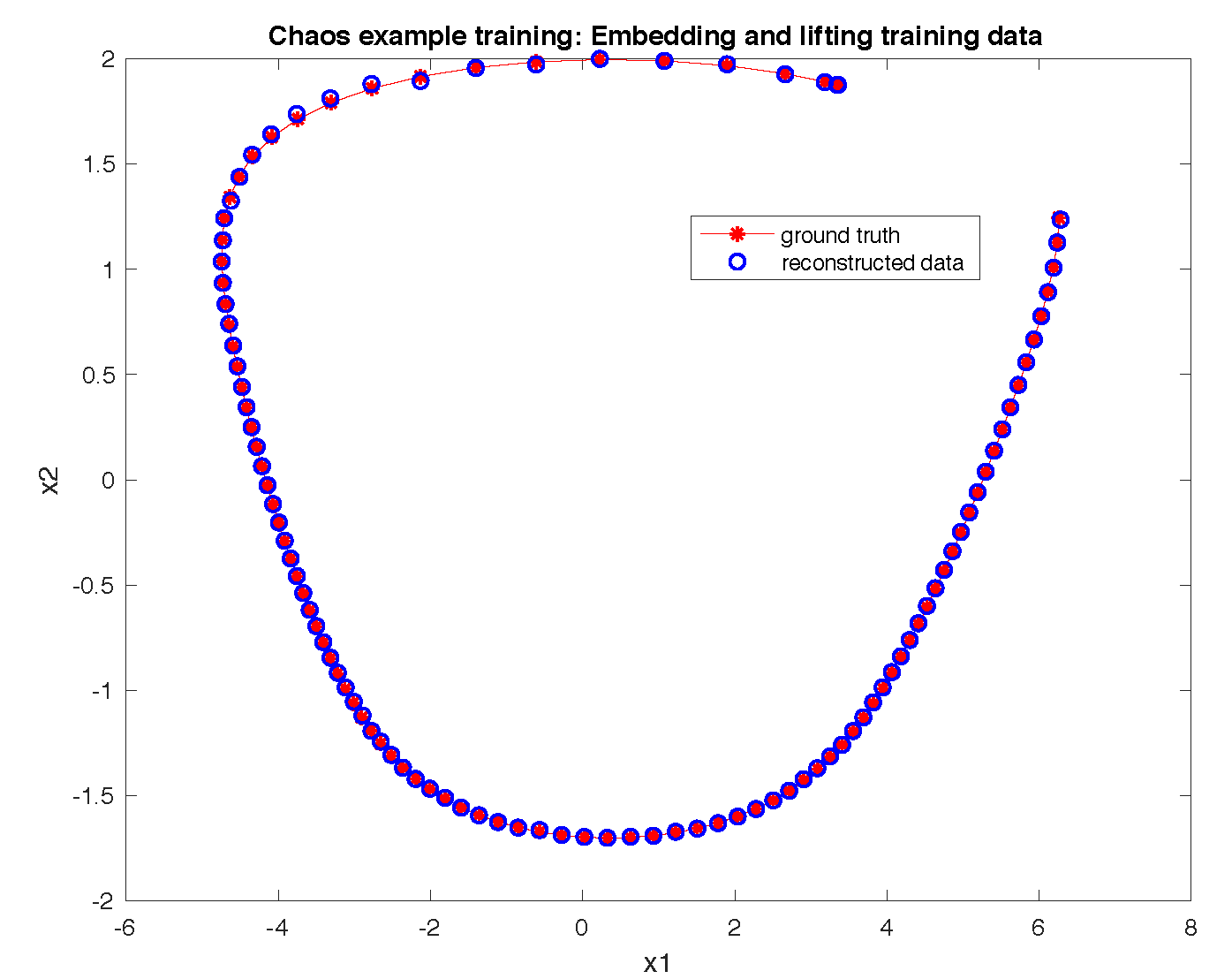}
\includegraphics[width=2.1 in]{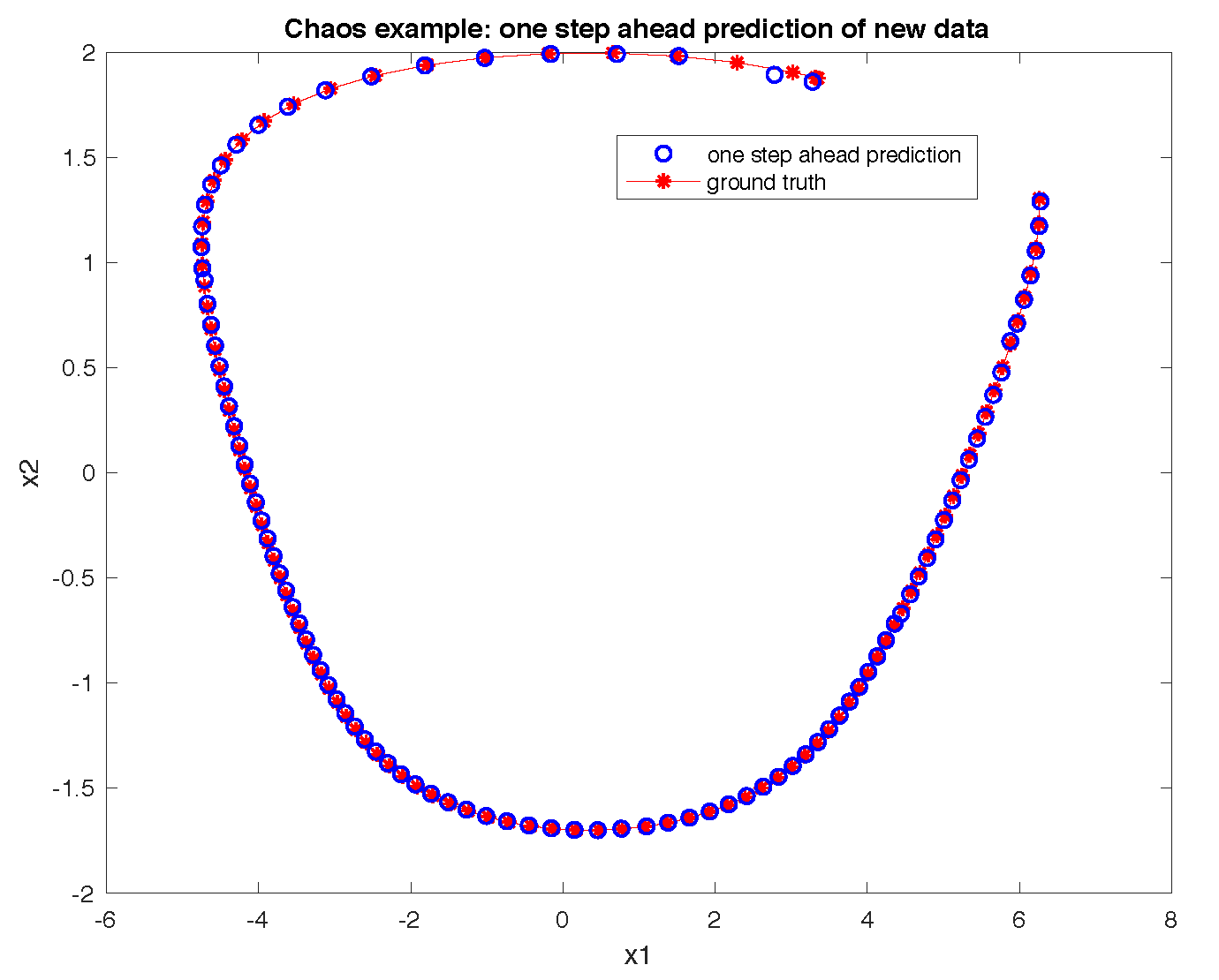}
\caption{\small Example 4: Left:  training data. Right: one step ahead predictions over an 100 step horizon, red: ground truth, blue: predictions using the proposed pipeline.}\label{fig:chaos2}
\end{center}

\end{figure}
\section{Conclusions}\label{sec:conclusions}

This paper proposes a convex optimization approach to learning Koopman operators from data. The main idea is to use delay coordinates and nonlinear, kernel based embeddings to recast the problem as a rank-constrained optimization. In turn, this optimization can be relaxed to a tractable semi-definite program.
Salient features of this approach are its ability to certify that the solution to this SDP indeed solves the original problem, and the fact that neither the order of the embedding nor of the dynamics governing their evolution need to be specified a-priori.  Further, by seeking embeddings that minimize the order of these dynamics, it leads to simpler models than those obtain for instance by simply factoring the Hankel matrix of the observed data.
The effectiveness of the proposed technique was illustrated with two examples that exhibit chaotic behavior.  In principle the approach proposed here requires solving a large SDP, and it is well known that SDPs have poor scaling properties.
However, as shown in the Appendix, the specific optimization arising in this paper exhibits an underlying sparse structure (chordal sparsity) than can be exploited to obtain algorithms whose complexity scales linearly with the number of data points, when these SDPs are solved using an ADMM based method such as the one proposed in  \cite{ZhengADMM2020}. This extension along with an extension to piecewise linear dynamics on the manifold, is currently being explored.

\bibliography{koopman.bib}
\appendix


\section{Technical Proofs}

\noindent \textbf{Proof of Theorem \ref{Teo:theo1}}. For simplicity assume that the roots of $\mathcal{P}(\rho)$ are simple.
Begin by noting that the $r^*$ linearly independent vectors $ \mat{v}_i$ are in $\mathcal{N}_R(\mat{p}^T)$, the right  null space of $\mat{p}^T$. Since dim$(\mathcal{N}_R(\mat{p}^T))=r^*$, it follows that these vectors form a basis of 
 $\mathcal{N}_R(\mat{p}^T)$, e.g.  $\mathcal{N}_R(\mat{p}^T) = \text{span}(\mat{V})$.  Next,  
  let ${y}_{k,j}^{(i)}$ denote the $j^{th}$ component of $\mat{y}_{k}^{(i)}$. 
  By construction the vectors 
 \[ \greekvect{\upsilon}_{k,j}^{(i)} \doteq  \begin{bmatrix} {y}_{k-r^*,j}^{(i)} & {y}_{k-r^*+1,j}^{(i)} & \ldots & {y}_{k,j}^{(i)} \end{bmatrix}^T \in \mathcal{N}_R(\mat{p}^T)  \]
 Hence $ \greekvect{\upsilon}_{k,j}^{(i)} \in \text{span}(\mat{V})$ and can be written as
 $\greekvect{\upsilon}_{k,j}^{(i)} = \mat{V}\mat{c}_{k,j}^{(i)}$.
  Repeating this reasoning for each component of $\mat{y}_k^{(i)}$ leads to:
    \[  \greekvect{\psi^{(i)}_{k}} = \left (\mat{V}\otimes \mat{I}_{m} \right ) \mat{c}_k^{(i)} 
  \doteq  \mat{D} \mat{c}_k^{(i)}\]
where $\mat{c}_k^{(i)} \doteq \sum_{j=1}^{m} \mat{c}_{k,j}^{(i)} \otimes \mat{e}_j$, $\mat{e}_j \doteq \begin{bmatrix} 0 \ldots 1 \ldots 0 \end{bmatrix}^T$ and where, for notational simplicity we use the shorthand $ \greekvect{\psi^{(i)}_{k}} \doteq \greekvect{\psi}(\greekvect{\xi}^{(i)}_k)$.
 Applying the same reasoning to each (block) component of  $\greekvect{\psi^{(i)}_{(k+1)}}$ yields:
    \[ \begin{aligned} \mat{y}_{\ell+1}^{(i)}& =
     ( \begin{bmatrix}\rho_1^{\ell+1} \ldots \rho_{r^*}^{\ell+1} \end{bmatrix} \otimes \mat{I}_m )\mat{c}^{(i)}_{k},  \; \ell=k-r^*,\ldots,k \\
    & =   \left( ( \begin{bmatrix}\rho_1^{\ell} \ldots \rho_{r^*}^{\ell} \end{bmatrix}  \text{diag}(\rho_i) ) \otimes \mat{I}_m \right ) \mat{c}^{(i)}_{k}\\
    & =  ( \begin{bmatrix}\rho_1^{\ell} \ldots \rho_{r^*}^{\ell} \end{bmatrix} \otimes \mat{I}_m )\greekvect{\Lambda} \mat{c}^{(i)}_{k}
  \end{aligned} 
     \]
     where $\greekvect{\Lambda} \doteq \text{diag}(\rho_i)\otimes \mat{I}_m $ and where we used the Kronecker's product property
   $(\mat{A}\mat{C})\otimes(\mat{B}\mat{D}) = (\mat{A}\otimes \mat{B})(\mat{C}\otimes \mat{D})$.  
       Thus
    \beq \label{eq:kron2} \begin{aligned}
       \greekvect{\psi^{(i)}_{k+1}} & = 
   \mat{D}\greekvect{\Lambda}\mat{c}^{(i)}_k \doteq  \mat{D}\mat{c}^{(i)}_{k+1}  \Rightarrow \mat{c}^{(i)}_{k+1} =\greekvect{\Lambda} \mat{c}^{(i)}_{k} 
  \end{aligned} 
  \eeq
      It follows that  $\greekvect{\Lambda} $  propagates $\mat{c}^{(i)}_{k} $, the coordinates of $\greekvect{\psi}^{(i)}(\greekvect{\xi}_k)$. Hence the eigenfunctions
       $\greekvect{\phi}_{j,\ell}(\greekvect{\xi})$  of the Koopman operator have the form $\greekvect{\phi}_{j,\ell}(\greekvect{\xi}) \doteq \mat{v}_j  \otimes  \mat{e}_{\ell}$,  $j=1,\ldots,r^*$, $\ell=1\ldots m$.  \qed

\noindent \textbf{Proof of Theorem \ref{Theo:SDP}}.    By construction, $(\mat{y}_s^{(i)})^T\mat{y}^{(i)}_t=\mat{K}_{s,t}^{(i)}$, hence satisfaction of \eqref{eq:Kc1}-\eqref{eq:Kc2} implies satisfaction of 
\eqref{eq:feas2}-\eqref{eq:feas3}. Consider now the corresponding Hankel matrices $\mat{H}_{\mat{y}^{(i)}}$.  From the definitions of $\mat{G}^{(i)}$ and $\vect{K}_{\ell,r}^{(i)}$ it follows that $\mat{H}^T_{\mat{y}^{(i)}}\mat{H}_{\mat{y}^{(i)}} = \mat{G}^{(i)}$. Since by construction rank($\mat{G}^{(i)})$ $\leq r^*$  then rank$(\mat{H}_{\mat{y}^{(i)}}) \leq r^* < r+1$ and hence \eqref{eq:feas1} is also satisfied. 
\qed

\noindent \textbf{Proof of Theorem \ref{teo:lowner}.} Existence of $g(.)$ follows from the properties of the Loewner matrix discussed in Section \ref{sec:Loewner}. Since by construction $\mat{K}_{\mat{x}} \succeq 0$, it follows that it defines a valid Kernel in the set $\mathcal{X}\cup \{ \mat{x}^*\}$. Further, since rank$(\mat{K_x}) \leq n$ it can be factored as $\mat{K}_{\mat{x}}=\mat{X}^T\mat{X}=\mat{K}_{\mat{x}}$, where, again by construction, $\mat{X}(i+1,:)=\mat{x}_i, i=1:|\mathcal{X}|$.  Let $\mat{x}^* = \mat{X}(1,:)$. From the definitions of $\greekvect{\kappa_{\mat{x}}},\greekvect{\kappa_{\mat{y}}} $ it follows that 
 $$\begin{aligned}
 & (\mat{x}^*)^T\mat{x}^* = \greekvect{\kappa_{\mat{x} _1}}, \; (\mat{y}^*)^T\mat{y}^* = \greekvect{\kappa_{\mat{y} _1}} \\
  &\mat{x}_i^T\mat{x}^* = \greekvect{\kappa_{\mat{x} _i}},  \; \mat{y}_i^T\mat{y}^* = \greekvect{\kappa_{\mat{y}_i}} \\
  &\mat{x}_i^T\mat{x}_i = \greekvect{\kappa_{\mat{x} _j}}, \; \mat{y}_i^T\mat{y}_i = \greekvect{\kappa_{\mat{y} _j}}
\end{aligned}$$  
 where $j$ is defined in \eqref{eq:SDPx4}.  Thus,  the constraint \eqref{eq:SDPx4} is simply a restatement of \eqref{eq:Kc1} in terms of the elements of  $\greekvect{\kappa_{\mat{x}}}, \greekvect{\kappa_{\mat{y}}}$.  \qed

\section{Exploting Chordal Sparsity}\label{sec:sparsity}

\subsection{Semi-Definite Programs and Rank Minimization Over Chordal Graphs}\label{sec:SDPsparsity}

\label{sec:chordal}

In this paper, we will reduce the problem of identifying Koopman operators to a constrained rank minimization of the form
\beq \label{eq:rc_SDP}\begin{split}
 & \min_{\mat{X}\succeq 0} \text{rank $(\mat{X})$ subject to } \\
& \text{Trace $(\mat{A}_i\mat{X}) \leq b_i, i=1,\ldots,n_c$,  $\mat{X},\mat{A}_i \in \mathbb{R}^{n\times n}$}
\end{split}
\eeq
In the specific problems arising in this paper only a small number of entries of  $\mat{X}$ appear in the trace constraints, while the role of the other entries is just to enforce that $\mat{X}\succeq 0$. Thus, as long as existence of a  \emph{minimum rank} PSD completion is guaranteed, these variables do not have to be explicitly found, allowing for a substantial computational complexity reduction.  Specifically, to the optimization \eqref{eq:rc_SDP} one can associate a graph $\mathcal{G}(\mathcal{V,E})$ with $n$ vertices in
$\mathcal{V}$  and edge set $\mathcal{E}$, where there is an edge between vertices $j$ and $\ell$ if the element $(j,\ell)$  of any of the matrices $\mat{A}_i$ is nonzero. Given a graph $\mathcal{G}$, define the cone 
\[\mathbb{S}_+^n (\mathcal{E},?)  \doteq \{ \mat{X} \in \mathbb{S}_+^n \colon \text{$\mat{X}_{i,j}$ given if $(i,j) \in \mathcal{E}$} \}\]
that is, the cone of matrices with entries fixed over the edges $\mathcal{E}$ than can be completed to be PSD.  When the graph $\mathcal{G}$ is chordal, the minimum rank over all possible matrix completions over this cone has an explicit expression, given by Dancis' Theorem:
\begin{Theorem} [\cite{Dancis92}] \label{T:DancisTheorem}
     Let $\mathcal{G}(\mathcal{V},\mathcal{E})$ be a chordal graph with a set of maximal cliques $\{\mathcal{C}_1,\mathcal{C}_2, \ldots, \mathcal{C}_{n_c}\}$. Then,
     for any  $\mat{X}\in\mathbb{S}^n_+(\mathcal{E},?)$ there exist at least one minimum rank PSD completion where 
     $$ \text{rank}(\mat{X}) = \max_{1 \leq k \leq n_c} \text{rank}(\mat{E}_{\mathcal{C}_k} \mat{X}\mat{E}_{\mathcal{C}_k}^T) $$
    where the $0/1$ matrix $\mat{E}_{\mathcal{C}_k}$ selects the variables of $\mat{X}$ corresponding to edges in the clique $\mathcal{C}_k$.
    \end{Theorem}
    In addition, the cone $\mathbb{S}_+^n (\mathcal{E},?)$ can be characterized using the following result (Grone's Theorem):
\begin{Theorem} [\cite{grone1984positive}] \label{T:GroneTheorem}
     Let $\mathcal{G}(\mathcal{V},\mathcal{E})$ be a chordal graph with a set of maximal cliques $\{\mathcal{C}_1,\mathcal{C}_2, \ldots, \mathcal{C}_{n_c}\}$. Then, $\mat{X}\in\mathbb{S}^n_+(\mathcal{E},?)$ if and only if
     $$ \mat{X}_k = \mat{E}_{\mathcal{C}_k} \mat{X}\mat{E}_{\mathcal{C}_k}^T \in \mathbb{S}_+,
    \qquad k=1,\,\ldots,\,n_c$$
     \end{Theorem}
Combining the two theorems above leads to the following result.
\begin{Corollary}\label{coro:chordal} The optimization
\eqref{eq:rc_SDP} is equivalent to:
\beq \label{eq:rc_chordal}\begin{split}
 & \min \sum_k \text{rank ($\mat{E}_{\mathcal{C}_k} \mat{X}\mat{E}_{\mathcal{C}_k}^T$) subject to } \\
& \mat{E}_{\mathcal{C}_k} \mat{X}\mat{E}_{\mathcal{C}_k}^T \succeq 0 \\
& \text{Trace $(\mat{A}_i\mat{X}) \leq b_i, i=1,\ldots,n_c$}
\end{split}
\eeq
\end{Corollary}

 Since rank minimization problems are generically  NP-hard, a standard convex relaxation is to replace rank by is convex envelope, trace
\cite{MohanFazel}. This substitution leads to a convex SDP that can be solved to  $\epsilon$--optimality in polynomial time using interior point (IP) methods. However, while efficient, these methods have relatively poor scaling properties ($\mathcal{O}(n^2n_c^2 + n^3n_c)$).  On the other hand,  as we will show in the sequel, the specific problem arising in this paper has chordal sparsity. Hence, the use of Corollary \ref{coro:chordal}  to decompose the objective into $\sum_k \text{Trace}(\mat{E}_{\mathcal{C}_k} \mat{X}\mat{E}_{\mathcal{C}_k}^T)$ leads to a SDP where 
 each of the PSD constraints has  (size of $\mathcal{C}_k$) variables. When combined with an ADMM approach where the cost of each iteration is 
 $\mathcal{O}(\text{number of variables}^3)$ \cite{ZhengADMM2020}, using this decomposition leads to a reduction of $\frac{1}{n_c}(\frac{n}{\text{max size of clique}})^3$ in computational complexity.

\subsection{Exploiting Chordal Sparsity}\label{sec:sparsity}

The approach outlined in Section \ref{sec:relaxation} works well for small to medium sized problems. However, its computational complexity grows roughly as
$\mathcal{O}(\text{(number of data points)}^6)$\footnote{The number of free variables in $\mat{K}$ is $(\text{number of data points)}^2$. Thus, if the SDP is solved using an interior point method,  computational complexity scales roughly as $\mathcal{O}(\text{(number of variables)}^6)$.}. Fortunately,  as we show next,  the  convex relaxation of \eqref{eq:Kc0}-\eqref{eq:Kc2} is endowed with chordal sparsity. Hence the  decomposition outlined in Section \ref{sec:chordal}, combined with an ADMM based algorithm such as the one proposed in  \cite{ZhengADMM2020}, can be exploited to substantially reduce computational complexity.
Let $r$ be an upper bound of the optimal rank $r^*$. 
Note that the only elements of $\mat{K}^{(i)}$ that appear explicitly in \eqref{eq:Kc0}-\eqref{eq:Kc2} are those of the form $\mat{K}^{(i)}_{s,t}$ where either
$|s-t|\leq r$ or  $\|\mat{x}_s - \mat{x}_r\|_2 \leq \delta$. Let 
\beq \label{eq:cliques} \begin{split}
&\mathcal{T}_{\ell} \doteq \{ (s,t): \ell \leq s,t \leq r \}\\
& \mathcal{S}_{\ell} \doteq \{(s,q): \text{$ \ell \leq s \leq r$ and $\|\mat{x}_s - \mat{x}_q \|_2 \leq \delta$} \} \\
& \mathcal{E}_\ell \doteq \mathcal{T}_\ell \cup \mathcal{S}_\ell
\end{split}
\eeq
To the  optimization problem  \eqref{eq:Kc0}-\eqref{eq:Kc2} we can associate  a graph with cliques  $\mathcal{C}_\ell$ defined by the edge sets $\mathcal{E}_\ell$.  From Corollary \ref{coro:chordal} and the fact that only the variables in $\mathcal{C}_\ell$ appear in the objective \eqref{eq:Kc0}, it follows that each constraint $\mat{K}^{(i)} \succeq 0$ can be replaced by a collection of smaller constraints of the form $\mat{E}^T_{\mathcal{C}_\ell}\mat{K}^{(i)}\mat{E}_{\mathcal{C}_\ell} \succeq 0$, where the matrix $\mat{E}_{\mathcal{C}_\ell}$ selects the entries of $\mat{K}^{(i)}$ corresponding to edges in
$\mathcal{C}_\ell$.  Assuming a fixed number $n_v$ of spacial neighbours,  the size of each clique $\mathcal{C}_\ell$  is given by $|\mathcal{C}_\ell | =(r+1)(1+n_v)$ and each trajectory has $T_i$ cliques. It follows that the computational complexity when using the clique-based decomposition roughly decreases by a factor of
$\frac{(1+r)^3(1+n_v)^3}{T_i}$. It is worth noting that, when using the clique decomposition, the overall computational complexity increases as $T_i 
(1+r)^3(1+n_v)^3$. This scaling is linear, rather than polynomial, in the number of data points.


\end{document}